\documentclass[12pt]{amsart}
\usepackage{amsmath,amsfonts,amssymb} \usepackage{color}
\usepackage[all,cmtip]{xy}
\usepackage{graphicx}
 \usepackage{youngtab}

 \newcommand{\bR}{\mathbb{R}}

\newcommand{\bT}{{\bf T}}

   \newcommand{\cE}{\mathcal{E}}

 \newcommand{\pd}{\partial}

\newcommand{\be}{\begin{equation}}
\newcommand{\ee}{\end{equation}}
\newcommand{\bea}{\begin{eqnarray}}
\newcommand{\eea}{\end{eqnarray}}
\newcommand{\ben}{\begin{eqnarray*}}
\newcommand{\een}{\end{eqnarray*}}

\newtheorem{cor}{Corollary}[section]

 \newtheorem{thm}[cor]{Theorem}
\theoremstyle{remark}

\definecolor{A}{rgb}{.75,1,.75}

\definecolor{green}{rgb}{0,1,0}
\definecolor{yellow}{rgb}{1,1,0}
\definecolor{orange}{rgb}{1,.7,0}
\definecolor{red}{rgb}{1,0,0}
\definecolor{white}{rgb}{1,1,1}

 \makeindex
\begin{document}
\title
{Information Theory and Statistical Mechanics Revisited}

\author{Jian Zhou}
\address{Department of Mathematical Sciences\\Tsinghua University\\Beijng, 100084, China}
\email{jzhou@math.tsinghua.edu.cn}

\begin{abstract}
We derive Bose-Einstein statistics and Fermi-Dirac statistics
by Principle of Maximum Entropy  applied to two families of entropy functions
different from the Boltzmann-Gibbs-Shannon entropy.
These entropy functions are identified with special cases of modified Naudts' $\phi$-entropy.
\end{abstract}

\maketitle

In a pioneering work sixty years ago,
Jaynes \cite{Jaynes} initiated the study of statistical physics
by information theory.
He expounded the Principle of Maximum Entropy and applied it
to Boltzmann-Gibbs-Shannon entropy to derive Boltzmann-Gibbs statisitics.
Many different entropy functionals have been introduced, studied and used in
many areas where statistics have been applied.
But in physics,
only the Boltzmann-Gibbs-Shanon entropy has been considered as ``physical".
Tsallis \cite{Tsallis} has proposed the q-entropy  as
another kind of physical entropy functional.
Naudts \cite{Naudts} introduced $\phi$-exponentials, $\phi$-logarithms
and $\phi$-entropy as a further generalization related to the $q$-entropy.
On \cite[p. 94]{Naudts} he wrote:``One of the conclusions is that the q-deformed exponential family occurs in
a natural way within the context of classical mechanics. The more abstract
generalisations discussed in the final chapters may seem less important from a
physics point of view. But they have been helpful in elucidating the structure
of the theory of generalised exponential families."
We will show that such generalizations are indeed of physical importance,
in particular,
in understanding the Bose-Einstein statistics and the Fermi-Dirac statistics
by Principle of Maximum Entropy.

We will first consider the subjective statistical physics of free bosons and free fermions.
More precisely,
we will derive Bose-Einstein statistics and Fermi-Dirac statistics
from the Principle of Maximum Entropy,
not for the Boltzmann-Gibbs-Shanon entropy as Jaynes did for the Boltzmann-Gibbs statistics,
but instead for two different entropy functions.
We quote here Tsallis \cite[p. 4]{Tsallis}: ``Indeed,
the physically important entropy - a crucial concept - is {\em not} thought as being an universal
functional that is given once for ever, but it rather is a delicate and powerful concept
to be carefully constructed for classes of systems."
We will actually consider two different families interpolating the Bose-Einstein
and Fermi-Dirac   statistics.
They give us two families of entropy functions.

Next we will present a unified understanding of the Boltzmann-Gibbs weight function,
the Bose-Einstein weight function and the Fermi-Dirac weight function
from the point of view of natural parameters of
exponential families.

Finally,
the three entropy functions and the three weight functions
are unified in terms of special cases of the generalized logarithmic functions
and generalized exponential functions developed by Naudts, respectively.
Some modifications are introduced for this purpose.

We also briefly treat the case of fractional exclusion statistics \cite{Haldane, Wu}.
In a subsequent work we will treat the case of general statistics
interpolating the Bose-Einstein and Fermi-Dirac statistics.

\section{Unified Derivation of Boltzmann-Gibbs Statistics, Bose-Einstein Statistics and Fermi-Dirac Statistics from Principle of Maximum Entropy}

In this section we will introduce the statistical manifolds
that describe a single particle, in finitely many states.
Then we will  use suitable entropy functions on these manifolds
and the Principle of Maximum Entropy to give a unified derivation of
three important physical statistics that describe noninteracting particles.

\subsection{The statistical manifold}

Suppose that one is performing a test with one observable $\cE$ with finitely many outcomes
$\{E_1, \dots, E_n\}$,
$E_1 < \cdots < E_n$.
Suppose that each outcome has a positive probability $p_i$ of appearance:
\be
p(\cE = E_i) = p_i, \;\;\; p_i > 0, \;\; i =1, \dots, n;
\ee
these probabilities are required to summed up to one:
\be \label{eqn:Constraint}
p_1 + \cdots + p_n = 1.
\ee
Such a distribution is called a {\em categorical distributions} in statistics.
Putting all the possible probability distributions together,
one gets an open $(n-1)$-simplex:
\be
P_n = \{(p_1, \dots, p_n) \in \bR^n\;|\; p_1 + \cdots + p_n =1, \; p_i > 0, \; i=1, \dots, n\}.
\ee
It is an open $(n-1)$-dimensional manifold.
We will take $p_1, \dots, p_{n-1}$ as coordinates on $P_n$,
and express $p_n$ as a function in these coordinates:
\be
p_n = 1- p_1 - \cdots - p_{n-1}.
\ee
As in Jaynes \cite{Jaynes}
one actually works with some submanifolds of $P_n$,
denoted by $P_n(E_1, \dots, E_n; E)$ and defined by the following constraint:
\be \label{eqn:EnergyConst}
p_1 E_1 + \cdots + p_n E_n = E,
\ee
where $E$ satisfies
\be
E_1 = \min\{E_1, \dots, E_n\} \leq E \leq \max\{E_1, \dots, E_n\} = E_n.
\ee

\subsection{The entropy functions}
\label{sec:Entropy}

Recall the Boltzmann-Gibbs-Shannon entropy function is defined by:
\be
H_{BGS} = - \sum_{i=1}^n p_i \log p_i.
\ee
We now introduce the following two entropy functions:
\bea
&& H_{BE} = \sum_{i=1}^n ((p_i +1)\ln( p_i+1) - p_i \ln p_i), \\
&& H_{FD} = \sum_{i=1}^n  ((p_i-1)\ln(1- p_i) - p_i \ln p_i).
\eea
The motivation for the introduction of these functions will be elaborated elsewhere.
Here let us just say $H_{BE}$ can be thought of as a discrete version
of \cite[(14)]{Yang-Yang}.

We also introduce a family of entropy functions:
\be
H_{\epsilon} = \sum_{i=1}^n (\frac{1}{\epsilon}(1  +\epsilon p_i)\ln(1+\epsilon p_i) - p_i \ln p_i).
\ee
Then we have
\begin{align}
H_1 & = H_{BE}, & H_0 & = H_{BGS} + 1, & H_{-1} & = H_{FD}.
\end{align}

\subsection{Principle of Maximum Entropy}
\label{sec:PME}

We now show that the application of Principle of Maximum Entropy
to the three entropy functions in last subsection on the
statistical manifold $P_n(E_1, \dots, E_n)$ gives us a unified
derivation of the Boltzmann-Gibbs statistics,
Bose-Einstein statistics, Fermi-Dirac statistics,
and more generally,
Acharya-Swamy statistics \cite{Ach-Swa}.

\begin{thm}
On $P_n(E_1, \dots, E_n; E)$,
the entropy function $H_{\epsilon}$ achieves its maximum at
the points:
\be \label{eqn:BE}
p_i(\epsilon) = \frac{1}{e^{a+bE_i}-\epsilon}, \;\; i=1, \dots, n
\ee
for some constants $a$ and $b$.
\end{thm}

\begin{proof}
By the method of Lagrange multiplier,
consider the function
\ben
F_{\epsilon} & = &  \sum_{i=1}^n (\frac{1}{\epsilon}(1+ \epsilon p_i)\ln(1 + \epsilon p_i) - p_i \ln p_i)\\
& + & a (1- \sum_{i=1}^n p_i) + b (E - \sum_{i=1}^n p_i E_i).
\een
One easily gets:
\be
\frac{\pd F_{\epsilon}}{\pd p_i} = \ln \frac{1 + \epsilon p_i}{p_i} - a  - b E_i, \;\; i=1, \dots, n,
\ee
so the critical point where $\frac{\pd F_{\epsilon}}{\pd p_i} = 0$ for all $i=1, \dots, n$
is given by
\be
p_i(\epsilon) = \frac{1}{e^{a+bE_i} - \epsilon}.
\ee
The entries of Hessian matrix of $F_{\epsilon}$ are given by:
\be
\frac{\pd^2F_{\epsilon}}{\pd p_i\pd p_j}
= -\delta_{ij} \frac{1}{p_i(1 +\epsilon p_i)},
\ee
the Hessian is clearly negatively definite.
\end{proof}

\section{Weight Functions as Inverse Functions of Natural Parameters of Exponential Families}

In this section we understand the weight functions:
\be
p_{\epsilon}(E) = \frac{1}{e^{a+b E} - \epsilon}
\ee
as inverse functions of natural parameters of exponential families.

\subsection{Exponential family}

In statistics,
an {\em exponential family} of probability densities is an n-dimensional model $S =
\{p_\theta \}$  of the form
\be \label{eqn:Exponential}
p(x; \theta) = \exp \biggl[C(x)+ \sum_{i=1}^n \theta^i T_i(x) - \psi(\theta) \biggr].
\ee
The parameters $\{\theta^i\}$ are called the {\em natural  parameters},
and the function $\psi(\theta)$ is determined by the normalization condition
\be
\int p(x; \theta) dx  = 1,
\ee
and so it is given by:
\be
\psi(\theta)
= \log \int \exp \biggl[C(x)+ \sum_{i=1}^n \theta^i T_i(x)  \biggr].
\ee

Recall when Gibbs \cite{Gibbs}
introduced the canonical ensemble in 1901
he postulated a distribution  of the form
\be
p(E) = \exp(G - \beta E)
\ee
where $G$ is a normalization constant and
where the control parameter $\beta = \frac{1}{kT}$ is the
inverse temperature.
This is an example of exponential family.

\subsection{Boltzmann-Gibbs weight functions as inverse natural parameters for the categorical distribution}

The categorical distribution is an example
of the exponential families.
First one can rewrite it in the following form:
\be
p= p_1^{[\cE=E_1]} \cdots p_n^{[\cE=E_n]},
\ee
where $[\cE = E_i]$ is the indicating function
that
equals to one when the energy level is $E_i$, zero otherwise.
Note
\be
\log p = \sum_{i=1}^n [\cE=E_i]  \ln p_i
= \sum_{i=1}^n [\cE=E_i] \cdot \ln p_i.
\ee
So by comparing with \eqref{eqn:Exponential},
one can take $T_i = [\cE = E_i]$,
and
the natural parameters can be taken to be:
\be
\eta_i = \ln p_i,
\ee
and so
\be
p_i = e^{\eta_i}.
\ee
This gives us the Boltzmann-Gibbs weight function
when we take $\eta_i = -(a+bE_i)$.

\subsection{Fermi-Dirac weight function as inverse natural parameter for the Bernoulli distribution}

The Fermi-Dirac weight function can be interpreted as the inverse function of natural
parameter of the Bernoulli distribution.
By Pauli's Exclusion Principle,
the outcome for observing a free fermion at a fixed state is like  the toss of coins,
it can be only be $0$ or $1$ particle at this state.
Suppose the probability is given by:
\begin{align}
p(X=1) & = p, &
p(X=0) & = 1- p.
\end{align}
The distribution can be written as
\be
P(X=x) = p^x (1-p)^{1-x}.
\ee
This is called the {\em Bernoulli distribution} in statistics.
This is also an example of exponential families:
\be
\log P = x \ln \frac{p}{1-p} + \ln (1-p).
\ee
The natural parameter is given by:
\be
\eta = \ln \frac{p}{1-p},
\ee
and so the inverse function is given by:
\be
p = \frac{1}{e^{-\eta} + 1}.
\ee
This is the Fermi-Dirac weight function
when we take $\eta = -(a+bE)$.

\subsection{The Acharya-Swamy weight function for $\epsilon < 0$ as inverse natural parameter for the Bernoulli distribution}

In the case of $\epsilon < 0$,
consider the probability distribution given by:
\be
P(X=x) 
= \biggl(\frac{p}{1+(1+\epsilon)p} \biggr)^x
\cdot \biggl( \frac{1 + \epsilon p}{1+(1+\epsilon)p}\biggr)^{1-x},
\ee
supported on the set $\{0,1\}$.
This is a curved Bernoulli distribution.
This is also an example of exponential families:
\be
\log P = x \ln \frac{p}{1+\epsilon p} + \ln \biggl( \frac{1 + \epsilon p}{1+(1+\epsilon)p}\biggr).
\ee
The natural parameter is given by:
\be
\eta = \ln \frac{p}{1+ \epsilon p},
\ee
and so the inverse function is given by:
\be
p = \frac{1}{e^{-\eta} - \epsilon}.
\ee
This gives the Acharya-Swamy weight function for $\epsilon < 0$
when we take $\eta = -(a+bE)$.

\subsection{The Bose-Einstein weight function as inverse natural parameter of the geometric distribution}

Similarly,
the Bose-Einstein weight function can be interpreted as the inverse function of natural
parameter of the geometric distribution.
The number of a free boson at a fixed state can be any nonnegative integer $n \geq 0$.
Suppose the probability is given by:
\be
P(X=n) = \frac{p^n}{(p+1)^{n+1}}, n =0, 1, 2, \dots.
\ee
This is called the {\em geometric distribution} in statistics.
Since
\be
\ln P(X =x) = \ln \frac{p^x}{(p+1)^{x+1}}
= x \ln \frac{p}{p + 1} + \ln \frac{1}{p+1},
\ee
one sees that it is an exponential family with natural parameter:
\be
\eta = \ln \frac{p}{p+1},
\ee
with inverse function:
\be
p = \frac{1}{e^{-\eta} - 1}.
\ee
This is the Bose-Einstein weight function
when we take $\eta = -(a+bE)$.

\subsection{The Acharya-Swamy weight function for $\epsilon > 0$ as inverse natural parameter for the geometric distribution}

Consider the probability distribution given by:
\be
P(X=n) = \frac{(p/(1+(\epsilon-1)p))^n}{((1 + \epsilon p)/(1+(\epsilon-1)p))^{n+1}}, n =0, 1, 2, \dots.
\ee
This is  a curved geometric distribution.
Since
\be
\ln P(X =x)
= x \ln \frac{p}{1+ \epsilon p} + \ln \frac{1+(\epsilon-1) p}{1+ \epsilon p},
\ee
one sees that it is an exponential family with natural parameter:
\be
\eta = \ln \frac{p}{1+ \epsilon p},
\ee
with inverse function:
\be
p = \frac{1}{e^{-\eta} - \epsilon}.
\ee
This is the Acharya-Swamy weight function for $\epsilon > 0$
when we take $\eta = -(a+bE)$.

\section{Bose-Einstein Statistics and Fermi-Dirac Statistics
as Generalized Statistical Physics}

The discussions of exponential families in last section serve as a psychological
vehicle that takes us to the notion of generalized exponential families
developed by Naudts \cite{Naudts},
which generalizes the $q$-exponential families of Tsallis \cite{Tsallis}.
We first recall the $\phi$-logarithm function,
the $\phi$-exponential function and the $\phi$-entropy function,
then we use their suitable modifications to study $H_\epsilon$ and $p_\epsilon$.

\subsection{The $\phi$-logarithm}

Fix a strictly positive non-decreasing function $\phi(u)$, defined on the positive
numbers $(0,+\infty)$.
It can be used to define a deformed logarithm by
\be
\ln_\phi(u) = \int_1^u dv \frac{1}{\phi(v)}, \;\;\; u > 0.
\ee
It satisfies $\ln_\phi(1) = 0$ and
\be
\frac{d}{du} \ln_\phi(u) = \frac{1}{\phi(u)}.
\ee
The natural logarithm is obtained with $\phi(u) = u$,
The Tsallis $q$-logarithm  is obtained with $\phi(u) = u^q$ for $q > 0$.

\subsection{The $\phi$-exponential function}

The inverse of the function $\ln_\phi(x)$ is called the
$\phi$-exponential and is denoted $\exp_\phi(x)$.
It can be written in terms of a function $\psi$ on $\bR$ defined by:
\be
\psi(u) = \begin{cases}
\phi( \exp_\phi(u)), & \text{if $u$ is in the range of $\ln_\phi$}, \\
0, &  \text{if $u$ is too small}, \\
+\infty, & \text{if $u$ is too large}.
\end{cases}
\ee
Clearly is $\phi(u) = \psi(\ln_\phi(u))$ for all $u > 0$.
Then $\exp_\phi$ is defined by:
\be
\exp_\phi(u) = 1 + \int_0^u dv \psi(v).
\ee
It is clear that $\exp_\phi(0) = 1$ and
\be
\frac{d}{du}\exp_\phi(u) = \psi(u).
\ee

\subsection{Deduced Logarithms}

The deduced logarithm is defined by
\be
\omega_\phi(u) = u \int_0^{1/u} dv \frac{v}{\phi(v)} - \int_0^1 dv \frac{v}{\phi(v)}
- \ln_\phi \frac{1}{u}.
\ee
It satisfies $\omega_\phi(1) = 0$ and that
\be
\frac{d}{du}
\omega_\phi (u)  =  \int_0^{1/u} dv \frac{v}{\phi(v)}.
\ee
Introduce a function:
\be
\chi(u) = \biggl[ \int_0^{1/u} dv \frac{v}{\phi(v)}\biggr]^{-1},
\ee
so one can see that the deduced logarithmic function is  the $\chi$-logarithm function:
\be
\omega_\phi (u) = \ln_\chi(u).
\ee

\subsection{The $\phi$-entropy}
The $\phi$-entropy is defined by \cite{Naudts}:
\be
H_\phi(p) = \sum_{i=1}^n  p_i \ln_\chi(1/p_i),
\ee
After a short calculation:
\be
H_\phi(p)
= - \sum_{i=1}^n  p_i \int_1^{p_i}
\frac{1}{v^2} \biggl[\int_0^{v} du \frac{u}{\phi(u)} \biggr] d v.
\ee

For our purpose,
we will define
\be
\tilde{H}_\phi(p)
= \sum_{i=1}^n  p_i \int_{p_i}^{+\infty}
\frac{1}{v^2} \biggl[\int_0^{v} du \frac{u}{\phi(u)} \biggr] d v
\ee
in order to remove some irrelevant constants.
We will call this the {\em modified $\phi$-entropy}.

\subsection{The entropy function $H_{BE}$ and $H_{FD}$
as modified $\phi$-entropy}

Define the following family of functions parameterized by $\epsilon$:
\be
\phi_\epsilon(p) = p(1+ \epsilon p).
\ee
We have
\ben
\tilde{H}_{\phi_\epsilon}(p)
& = & \sum_{i=1}^n p_i \int_{p_i}^{\infty} \frac{1}{v^2}
\biggl[\int_0^{v} du \frac{u}{u(1+\epsilon u)} \biggr] d v \\
& = & \sum_{i=1}^n p_i \int_{p_i}^\infty \frac{1}{\epsilon v^2} \ln |1 + \epsilon v|  d v \\
& = & \sum_{i=1}^n \biggl(\frac{1 + \epsilon p_i}{\epsilon} \ln (1 + \epsilon p_i)
- p_i \ln p_i \biggr) \\
& = & H_\epsilon.
\een
In particular, the entropy functions $H_{BE}$, $H_{BGS}+1$ and $H_{FD}$ are the modified
$\phi_{\epsilon}$-entropy functions for $\epsilon = +1$, $0$ and $-1$ respectively.

\subsection{Bose-Einstein weight function and Fermi-Dirac weight function as
modified $\phi$-exponential function}

Similarly,
we defined the modified $\phi$-logarithm function by:
\be
\widetilde{\ln}_\phi(u) = - \int_u^\infty dv \frac{1}{\phi(v)}, \;\;\; u > 0,
\ee
and define the modified $\phi$-exponential function $\widetilde{\exp}_\phi$ as its inverse function.

For the function
$\phi_\epsilon(p) = p(1+ \epsilon p)$,
we have
\be
\widetilde{\ln}_{\phi_\epsilon}(u) = - \int_u^\infty dv \frac{1}{v(1+ \epsilon v)}
= \ln|1+\epsilon u| -\ln(u).
\ee
It follows that
\be
u = \widetilde{\exp}_{\phi_\epsilon}(\eta) = \frac{1}{e^{\eta} - \epsilon},
\ee
and so we have
\be
p_\epsilon(E) = \widetilde{\exp}_{\phi_\epsilon}(a+ bE).
\ee

\section{Fractional Exclusion Statistics}

In this section,
we treat the case of fractional exclusion statistics of Haldane \cite{Haldane}.
We refer to \cite[Chapter 5]{Khare} for backgrounds.
Since the ideas are similar,
we will be very brief.

Wu \cite{Wu} has derived the following formula for the weight function:
\be
p(g) = \frac{1}{\omega(\eta) + g},
\ee
where the function $\omega(\eta)$ satisfies the functional equation:
\be
\omega(\eta)^g(1+\omega(\eta))^{1-g} = e^{\eta}.
\ee
For the special cases of $g = 0$ and $1$ we have $w(\eta) = e^{-\eta} - 1$
and $w(\eta) = e^{-\eta}$, and so we  recover the Bose-Einstein and the
Fermi-Dirac statistics respectively for $\eta = - (a + b E)$.
This weight function can be derived by maximizing the following family of entropy functions
parameterized by $g$:
\be
H_g(p) = (1+(1-g)p)\ln(1+(1-g)p) - (1-gp)\ln (1-gp) - p \ln p,
\ee
under the constraints \eqref{eqn:Constraint} and \eqref{eqn:EnergyConst}.
This is because
\be
\frac{\pd H_g}{\pd p} = \ln \frac{(1 + (1-g) p)^{1-g} (1- gp)^g}{p},
\ee
and so by the method of  Lagrange multiplier one can get:
\be
\frac{(1 + (1-g) p)^{1-g} (1- gp)^g}{p} = e^{-(a + b E)}.
\ee
One can readily check that
\bea
&& H_g = \tilde{H}_{\phi_g}(p), \\
&& p_g(E) = \widetilde{\exp}_{\phi_g}(a+ bE),
\eea
for the following function:
\be
\phi_g(p) = p(1-gp)(1+(1-g)p).
\ee

\section{Conclusions and Prospects}

In this paper we have generalized Jaynes' derivation of Boltzmann-Gibbs statistics
by the Principle of Maximum Entropy.
A family $H_\epsilon$ of entropy functions has been introduced to
give a unified derivation of Bose-Einstein, Boltzmann-Gibbs and  Fermi-Dirac statistics
together with the interpolating Acharya-Swamy statistics.
The family $H_\epsilon$ turns out to be a special case of Naudts' $\phi$-entropy
and the probabilities are $\phi$-exponentials,
with suitable modifications,
for $\phi$ given by
$\phi_\epsilon(p) = p + \epsilon p^2$.
A different interpolation of Bose-Einstein and  Fermi-Dirac statistics
is given by the $\phi_g$-exponential function and the corresponding
entropy function is given by the $\phi_g$-logarithm function,
for $\phi_g(p) = p(1-gp)(1+(1-g)p)$.
The two series of functions $\phi_\epsilon$ and $\phi_g$ suggest us
to study  more general deformation of $\phi=p$ given by
$
\phi_\bT(p) = p - \sum_{n \geq 2} T_{n-1} p^n.
$
In a subsequent work we will verify that other statistics interpolating Bose-Einstein statistics
and Fermi-Dirac statistics are $\phi_\bT$-exponential functions and
are critical point of $\phi_\bT$-entropy functions.
Furthermore,
any deformation of the Boltzmann-Gibbs-Shannon entropy in some suitable sense
can be obtained as a $\phi_\bT$-entropy.
We will use such considerations to establish a connection with string theory.
More precisely,
we will show that some computations in string theory can be used to generate interpolating
statistics.

\end{document}